\documentclass[preprint,5p,12pt,twocolumn,fleqn]{elsarticle}

\PassOptionsToPackage{hyphens}{url}
\usepackage{hyperref}
\usepackage[table]{xcolor}

\usepackage{amssymb,amsmath,colonequals,xspace,tikz}
\usetikzlibrary{calc,decorations.pathmorphing,fit,backgrounds,positioning,arrows,shapes}
\tikzstyle{knoten}=[circle,draw=black,thin,fill=white,inner sep=0pt,minimum size=4.5mm]
\tikzstyle{knotenklein}=[circle,draw=black,thin,fill=white,inner sep=0pt,minimum size=2.5mm]

\newtheorem{theorem}{Theorem}
\newtheorem{lemma}{Lemma}
\newtheorem{corollary}{Corollary}
\newproof{proof}{Proof}
\newdefinition{definition}{Definition}
\newdefinition{assumption}{Assumption}

\renewcommand{\o}[1]{\overline{#1}}
\newcommand{\CMCFPC}{$\textsc{BCMCFP}_{\mathbb{R}}$\xspace}

\makeatletter
\def\NAT@spacechar{~}
\makeatother

\let\turc\c
\usepackage{etoolbox}
\robustify\turc

\renewcommand{\c}{\ensuremath{c^\pi}}
\renewcommand{\b}{\ensuremath{b^\mu}}
\renewcommand{\d}{\ensuremath{d^{\pi,\mu}}}
\newcommand{\eb}{\ensuremath{{\o{e}}}}

\usepackage{eucal}

\makeatletter
\def\ps@pprintTitle{%
 \let\@oddhead\@empty
 \let\@evenhead\@empty
 \def\@oddfoot{}%
 \let\@evenfoot\@oddfoot}
\makeatother

\begin{document}
	
	\begin{frontmatter}
		\title{A Network Simplex Method for the Budget-Constrained Minimum Cost Flow Problem}
		
		\author[TUKL]{Michael Holzhauser\corref{cor1}}
		\ead{holzhauser@mathematik.uni-kl.de}

		\author[TUKL]{Sven O. Krumke}
		\ead{krumke@mathematik.uni-kl.de}

		\author[TUKL]{Clemens Thielen}
		\ead{thielen@mathematik.uni-kl.de}

		\cortext[cor1]{Corresponding author. Fax: +49 (631) 205-4737. Phone: +49 (631) 205-2511}

		\address[TUKL]{University of Kaiserslautern, Department of Mathematics\\Paul-Ehrlich-Str.~14, D-67663~Kaiserslautern, Germany}

		\begin{abstract}
			We present a specialized network simplex algorithm for the \emph{budget-constrained minimum cost flow problem}, which is an extension of the traditional minimum cost flow problem by a second kind of costs associated with each edge, whose total value in a feasible flow is constrained by a given budget~$B$. We present a fully combinatorial description of the algorithm that is based on a novel incorporation of two kinds of integral node potentials and three kinds of reduced costs. We prove optimality criteria and combine two methods that are commonly used to avoid cycling in traditional network simplex algorithms into new techniques that are applicable to our problem. With these techniques and our definition of the reduced costs, we are able to prove a pseudo-polynomial running time of the overall procedure, which can be further improved by incorporating Dantzig's pivoting rule. Moreover, we present computational results that compare our procedure with Gurobi \citep{Gurobi}.
		\end{abstract}

	\end{frontmatter}

	\section{Introduction}

	In this paper, we present a specialized network simplex algorithm for the \emph{budget-constrained minimum cost flow problem}. This problem embodies a natural extension of the traditional minimum cost flow problem (cf., e.g., \citep{Ahuja}) by a second kind of costs, called \emph{usage fees}, which are linear in the flow on the corresponding edge and bounded by a given budget~$B$. This extension allows us to solve many related problems such as the budget-constrained maximum dynamic flow problem (since each dynamic flow can be represented as a traditional minimum cost flow (cf. \citep{FordFulkersonDynamicFlow})) or the application of the $\varepsilon$-constraint method to the bicriteria minimum cost flow problem (cf., e.g., \citep{ChankongHaimes}).

	Since it was published by \citet{DantzigNetworkSimplex} (originally designed for the transportation problem), the network simplex algorithm for the traditional minimum cost flow problem has been improved progressively and is widely believed to be one of the most efficient solution methods for the minimum cost flow problem at present (cf. \citep[pp. 451--452]{Ahuja}, \citep{KovacsMinCostFlowEvaluation}). Simplex-type methods usually have to cope with the risk to ``get stuck'' in an infinite loop with no progress -- an effect that is referred to as \emph{cycling}. However, \citet{CunninghamStronglyFeasible} introduced the notion of \emph{strongly feasible bases} that may be used to prevent cycling. While the sequence of operations with no progress may still be exponentially large, several authors such as \citet{CunninghamStalling} and \citet{AhujaOrlinOnDegeneratePivots} provide measurements to keep the length of such sequences polynomially bounded. At present, the network simplex algorithm with the best time complexity is due to \citet{OrlinNetworkSimplexStronglyPolynomial} in combination with the dynamic tree data structure due to \citet{TarjanNetworkSimplexStronglyPolynomial} and achieves a running time of $\mathcal{O}(nm \log n \min\{\log n\mathcal{C}, m \log n\})$. Recent experimental results and overviews about algorithms for solving the standard minimum cost flow problem can be found in \citep{Ahuja,KovacsMinCostFlowEvaluation,SifalerasMinimumCostNetworkFlows}.

	To the best of our knowledge, the budget-constrained minimum cost flow problem was first investigated by \citet{GloverConstrainedTransshipment}, who investigated a network simplex algorithm for singly constrained transshipment problems. The authors use the observation that the definition of a \emph{basis structure} or \emph{spanning tree structure} as it is used in network simplex methods (cf., e.g., \citep[p. 408]{Ahuja}) can be extended to the case of the budget-constrained minimum cost flow problem. This fact was later also mentioned by \citet[p. 460]{Ahuja}, which led to the development of the algorithm presented in this paper. \citet{MathiesConstrainedNetworkFlowProblems} study a combination of a specialized network simplex algorithm and the Lagrangian relaxation method applied to network flow problems with multiple constraints. \citet{SpaeltiSatellitePlacement} show that the satellite placement problem contains a special case of constrained network flow problems and develop a specialized network simplex algorithm.

	Combinatorial algorithms for the model that we use here were investigated in \citet{BudgetConstrainedMinCostFlows}, where a strongly polynomial-time algorithm based on the interpretation of the problem as a bicriteria minimum cost flow problem was derived. In \citet{BudgetConstrainedComplexityApproximability}, the authors further presented an efficient weakly polynomial-time combinatorial algorithm that performs worse only by a logarithmic factor than the best algorithm for the traditional minimum cost flow problem as well as fully polynomial-time approximation schemes for the problem.

	The special case of \emph{budget-constrained transportation problems} was studied by \citet{KlingmanConstrainedTransportation,KlingmanSinglyConstrainedTransportation}, who present specialized network simplex methods for the problem. The related \emph{budget-constrained maximum flow problem} was first studied by \citet{AhujaConstrainedMaxFlow}, who present a weakly polynomial-time algorithm for the problem that is based on a capacity scaling variant of the successive shortest path algorithm. \citet{CalicskanAhuja} later showed that this algorithm may not return a feasible solution in a specific special case and presented a corrected version of the algorithm. The same author also presented a double scaling algorithm, a network simplex algorithm, and a cost scaling algorithm for the budget-constrained maximum flow problem and evaluated their empirical performance (cf. \citep{CalicskanDoubleScaling,CalicskanNetworkSimplex,CalicskanCostScaling}). In particular, he could show that the cost scaling variant outperforms the other two scaling variants (including the capacity scaling algorithm of \citet{AhujaConstrainedMaxFlow}) and that the network simplex algorithm clearly outperforms all known algorithms for the problem (including CPLEX applied to the LP formulation). \citet{KrumkeCapacityUpgrade} study the problem of finding a maximum flow in the case that the capacity of each edge can be improved using a given budget.

	The specialized network simplex algorithm published by \citet{CalicskanNetworkSimplex} was developed independently of the work in this paper and uses similar ideas to the presented ones. However, there are theoretical gaps left that we close with this paper. In particular, we present a fully combinatorial description of the algorithm and investigate its theoretical running time with respect to the more general minimum cost flow variant of the problem. We prove optimality criteria that are based on \emph{two} different kinds of \emph{integral} node potentials and \emph{three} kinds of reduced costs. Moreover, we provide comprehensive techniques to prevent cycling that combine two common methods used in the case of the network simplex algorithm for the traditional minimum cost flow problem in a novel way. These techniques and the integrality of the node potentials are the key elements needed to prove a (pseudo-polynomial) running time of our procedure, which embodies the first proven running time for a network simplex algorithm both for the budget-constrained maximum flow and minimum cost flow problem. As it will be shown, this running time can be further improved by incorporating Dantzig's pivoting rule for choosing the edge that enters the basis. Finally, we briefly discuss the empirical performance of our method and provide experimental results that compare it both to the algorithm presented by \citet{CalicskanNetworkSimplex} and to the Gurobi Solver \citep{Gurobi}.

	\section{Notation and Definitions}

	In the budget-constrained minimum cost flow problem (abbreviated as \emph{\CMCFPC} in the following), we are given a directed multigraph~$G=(V,E)$ with edge capacities~$u_e \in \mathbb{N}_{\geq 0}$, costs~$c_e \in \mathbb{Z}$ (i.e., we allow integral costs with arbitrary sign), and usage fees~$b_e \in \mathbb{N}_{\geq 0}$ per unit of flow on the edges~$e \in E$, as well as a budget~$B \in \mathbb{N}_{\geq 0}$. The aim is to find a feasible flow~$x$ in $G$ that minimizes~$\sum_{e \in E} c_e \cdot x_e$ subject to the budget-constraint~$\sum_{e \in E} b_e \cdot x_e \leq B$. In particular, we do not assume supplies and demands to exist, but stick to an equivalent circulation based formulation in which the excess is required to be zero at each node in $V$. The problem~\CMCFPC can be stated as a linear program as follows:
	{\allowdisplaybreaks
	\begin{subequations}\label{eqn:CMCFP:LP}
	\begin{align}
		\min		&	\sum_{e \in E} c_e \cdot x_e \label{eqn:CMCFP:LP_Objective} \\
		\text{s.t.}	&	\sum_{e \in \delta^-(v)} x_e - \sum_{e \in \delta^+(v)} x_e = 0 \ \forall v \in V, \\
	                &	\sum_{e \in E} b_e \cdot x_e \leq B,	\label{eqn:CMCFP:LP_Budget} \\
					&	0 \leq x_e \leq u_e \ \ \forall e \in E.
	\end{align}
	\end{subequations}
	}%
	Here, we denote by $\delta^+(v)$ ($\delta^-(v)$) the set of \emph{outgoing} (\emph{incoming}) edges of some node~$v \in V$.
	The following assumption can be easily established by adding artificial edges with large costs and usage fees:

	\begin{assumption}\label{ass:CMCFPC:Simplex:StronglyConnected}
		The underlying graph~$G$ is strongly connected.
	\end{assumption}

	One consequence of Assumption~\ref{ass:CMCFPC:Simplex:StronglyConnected} is that we can assume that $n \in \mathcal{O}(m)$ holds in the following. Furthermore, note that, since the zero-flow is always feasible and has objective value zero in our model, the optimal objective value of each instance of \CMCFPC is always non-positive. The problem is a generalization of the problem variant in which a desired flow value~$F$ is given since we can ``enforce'' such a flow value by adding an edge with negative costs of large absolute value and capacity~$F$ (cf. \citep{BudgetConstrainedMinCostFlows} for further details).

	Consider a traditional minimum cost flow with respect to the costs~$c$ that is computed by some state-of-the-art algorithm, for example the \emph{enhanced capacity scaling algorithm} by \citet{OrlinEnhancedCapacityScaling} running in $\mathcal{O}(m \log n \cdot (m + n\log n))$~time. If the total usage fee of this flow fulfills $\sum_{e \in E} b_e \cdot x_e \leq B$, we have clearly found an optimal solution to the given instance of the budget-constrained minimum cost flow problem and are done. We are particularly interested in the converse case that the budget is exceeded for this flow, which we will assume in the following. Note that the usage fee amounts to at least $B+1$ in this case since the computed traditional minimum cost flow is integral without loss of generality (cf., e.g., \citep[p. 449]{Ahuja}) and since the usage fees are integral as well:

	\begin{assumption}\label{ass:CMCFPC:Simplex:InfeasibleMinCostFlow}
		There is a minimum cost flow~$x$ with respect to the costs~$c$ that fulfills $\sum_{e \in E} b_e \cdot x_e \geq B + 1$.
	\end{assumption}

	In the following, we give insights into the notion of \emph{basis structures} in the context of \CMCFPC. In contrast to the network simplex algorithm for the traditional minimum cost flow problem (cf. \citep[pp. 405--407]{Ahuja}), we need to drop the assumption that the subgraph that is induced by basic edges is cycle free. Instead, the basis contains a cycle with non-zero usage fees, as it will be shown in detail in the following.

	Consider an edge~$\eb \in E$ and a partition of the remaining edges in $E \setminus \{\eb\}$ into three sets~$L$, $T$, and $U$. Let the edges in $T$ form a spanning-tree of the underlying graph~$G$. Since there is a unique (undirected) path between any two nodes in the subgraph induced by~$T$, each edge~$e \in L \cup U \cup \{\eb\}$ closes a unique (undirected) cycle~$C(e)$ together with the edges in $T$. In the following, for each $e \in L \cup U \cup \{\eb\}$, let $C^+(e)$ ($C^-(e)$) denote the set of edges that are oriented in the same (opposite) direction as $e$ in $C(e)$. The costs and usage fees of this cycle are then given by $c(C(e)) \colonequals \sum_{e \in C^+(e)} c_e - \sum_{e \in C^-(e)} c_e$ and $b(C(e)) \colonequals \sum_{e \in C^+(e)} b_e - \sum_{e \in C^-(e)} b_e$, respectively.


	In particular, note that the subgraph induced by the edges in $T \cup \{\eb\}$ contains a cycle~$C(\eb)$. We call such a tuple~$(L,T,U,\eb)$ a \emph{basis structure} of the budget-constrained minimum cost flow problem if $b(C(\eb)) \neq 0$. For a given basis structure~$(L,T,U,\eb)$, let $x$ denote a flow that fulfills~$x_e = 0$ for each $e \in L$, $x_e = u_e$ for each $e \in U$, and $b(x) \colonequals \sum_{e \in E} b_e \cdot x_e = B$ while maintaining flow conservation at each node~$v \in V$. We refer to $x$ as the \emph{basic solution} corresponding to $(L,T,U,\eb)$. We use the term \emph{basis structure} instead of just \emph{basis} in the following in order to emphasize that the tree and the edge~$\eb$ do not uniquely determine the corresponding basic solution (cf. \citet[Chapter 11]{Ahuja}).

	As shown in \citep{CalicskanNetworkSimplex}, the basic solution of each basis structure is uniquely defined and can be obtained in $\mathcal{O}(m)$~time: In a first step, we let $x_\eb = 0$ and determine the values of all edges in $T$ as it is done in the traditional network simplex method (cf. \citep[pp. 413--415]{Ahuja}). In a second step, the flow on the cycle~$C(\eb)$ is then increased (or decreased) until $b(x) = B$. In the following, we refer to a basis structure~$(L,T,U,\eb)$ as \emph{feasible} if the corresponding basic solution~$x$ is a feasible flow. In this case, we also refer to the flow~$x$ as a \emph{basic feasible flow}. The described procedure yields the following corollary:

	\begin{corollary}\label{cor:CMCFPC:Simplex:IfFractionalThenOnCycle}
		For each feasible basis structure~$(L,T,U,\eb)$ and its corresponding basic feasible flow~$x$, it holds that $x$ can be decomposed into two flows $x^I$ and $x^C$ such that $x = x^I + x^C$, where $x^I$ is integral and $x^C$ is non-zero (and possibly fractional) only on the edges of $C(\eb)$. \qed
	\end{corollary}

	Note that Corollary~\ref{cor:CMCFPC:Simplex:IfFractionalThenOnCycle} holds independently of whether the budget~$B$ is integral or not (this will be important in Section~\ref{sec:CMCFPC:Simplex:Termination}).

	As it is well known, we are able to restrict our considerations to such feasible basis structures and their corresponding basic feasible flows (cf. \citep[p. 460]{Ahuja}, \citep{GloverConstrainedTransshipment}):

	\begin{theorem}\label{thm:CMCFPC:Simplex:OptimalBasicFeasibleFlow}
		For each instance of \CMCFPC, there always exists an optimal flow that is basic feasible. \qed
	\end{theorem}

	As in the traditional network simplex algorithm, we associate \emph{node potentials} with each node~$v \in V$ in order to be able to check optimality quickly. However, since we are dealing with two kinds of costs, we maintain \emph{two} different node potentials~$\pi$ and $\mu$ that are defined with respect to the edge costs~$c_e$ and the usage fees~$b_e$, respectively. In particular, we define $\pi_{v_r} \colonequals 0$ and $\mu_{v_r} \colonequals 0$ for some arbitrary but fixed node~$v_r \in V$, which we select as the root of the spanning tree. We choose the node potentials~$\pi$ and $\mu$ in a way such that the \emph{reduced costs} $\c_e \colonequals c_e - \pi_v + \pi_w$ and $\b_e \colonequals b_e - \mu_v + \mu_w$ are zero for each edge~$e=(v,w) \in T$. With this restriction, the node potentials at each node~$v \in V$ are uniquely defined and can be computed in $\mathcal{O}(n)$~time (cf. \citep[pp. 411--412]{Ahuja} for further details).

	Note that the costs~$c(C)$ and usage fees~$b(c)$ of any cycle~$C$ equal its reduced costs~$\c(C) \colonequals \sum_{e \in C} \c_e$ and reduced usage fees~$\b(C) \colonequals \sum_{e \in C} \b_e$, respectively, since the node potentials cancel out in a circular fashion (cf. \citep[pp. 43--44]{Ahuja}). Consequently, since the reduced costs are zero for each edge in $T$, it holds that $c(C(e)) = \c_e$ and $b(C(e)) = \b_e$ for each cycle~$C(e)$ that is closed by some edge~$e \notin T$.


	For a given basis structure~$(L,T,U,\eb)$, we additionally assign a third kind of reduced costs~$\d_e$ to each edge~$e \in E$ in order to be able to decide if a basic feasible flow is optimal or to detect an edge that is able to improve the objective function\footnote{Remember that $\b_\eb = b(C(\eb)) \neq 0$ in any basis structure}:
	\begin{align}
		\d_e \colonequals \c_e - \c_\eb \cdot \frac{\b_e}{\b_\eb}.  \label{eqn:CMCFPC:Simplex:DefReducedCosts}
	\end{align}
	Intuitively, the reduced costs~$\d_e$ describe the effect that an increase of the flow on $C(e)$ by one unit and a decrease of the flow on $C(\eb)$ by $\frac{\b_e}{\b_\eb}$ units has on the objective function value. This will be shown in the following section. Note that $\d_e = 0$ for each $e \in T \cup \{\eb\}$ since $\c_e = 0$ and $\b_e = 0$ for each $e \in T$ and $\d_\eb = \c_\eb - \c_\eb \cdot \frac{\b_\eb}{\b_\eb} = 0$.

	\section{Network Simplex Pivots}
	\label{sec:CMCFPC:Simplex:Pivots}

	Before we describe the network simplex pivot in the case of \CMCFPC, it is useful to first recall the basic outline of the corresponding procedure in the case of the traditional network simplex algorithm: For a given basis structure~$(L,T,U)$ that consists of a set of edges~$L$ at their lower bound, a spanning tree~$T$, and a set of edges~$U$ at their upper bound, assume that there is an edge~$e \in L$ with negative reduced costs. Adding this \emph{entering edge}~$e$ to the spanning tree~$T$ closes a unique cycle~$C(e)$ with negative costs. By sending flow on $C(e)$ in the direction of $e$, we can, thus, improve the objective function value until, for some flow value~$\delta$, some \emph{leaving edge}~$e' \in C(e)$ reaches its lower or upper bound. By assigning this edge~$e'$ to $L$ or $U$, respectively, we obtain a new basis structure. This operation (adding an edge to $T$, sending flow on the cycle, removing one edge from the cycle) is called a \emph{simplex pivot}. One distinguishes between a \emph{degenerate simplex pivot} if $\delta = 0$ and a \emph{non-degenerate simplex pivot} if $\delta > 0$. Note that the objective function value does not increase during a simplex pivot, but only decreases strictly in the case of a non-degenerate simplex pivot.

	Now, for a given instance of \CMCFPC, let $(L,T,U,\eb)$ and $x$ denote a feasible basis structure and its basic feasible flow, respectively, and let $\pi$ and $\mu$ denote the corresponding node potentials. Assume that there is an edge~$e \in L$ with negative reduced costs~$\d_e < 0$. We show that we do not increase the objective function value if we add the then called \emph{entering edge}~$e$ to $T$ (which closes a new cycle~$C(e)$ together with the edges in $T$) and by sending suitable amounts of flow on \emph{both} of the cycles~$C(e)$ and $C(\eb)$ until the flow value on at least one \emph{leaving edge}~$e' \in T \cup \{e, \eb\}$ becomes equal to zero or $u_{e'}$. In this case, we can obtain a new basis structure~$(L',T',U',\eb')$ and continue the procedure.

	For some value~$\delta \geq 0$, let $x'$ be the flow defined as
	\begin{align}
		x' \colonequals x + \delta \cdot \chi(C(e)) - \delta \cdot \frac{\b_e}{\b_\eb} \cdot \chi(C(\eb)),   \label{eqn:CMCFPC:Simplex:SimplexUpdateL}
	\end{align}
	where, for any cycle~$C$ with forward edges~$C^+$ and backward edges~$C^-$, the flow $\chi(C)$ is defined as
	\begin{align*}
		(\chi(C))_e \colonequals \begin{cases}
			1, & \text{if } e \in C^+,\\
			-1, & \text{if } e \in C^-, \\
			0, & \text{else.}
		\end{cases}
	\end{align*}
	The new flow~$x'$ fulfills $b(x') = B$, since
	\begin{align*}
		b(x')	&= b(x) + \delta \cdot b(\chi(C(e))) \\
				&\hspace{0.5cm} - \delta \cdot \frac{\b_e}{\b_\eb} \cdot b(\chi(C(\eb))) \\
				&= b(x) + \delta \cdot b(C(e)) - \delta \cdot \frac{\b_e}{\b_\eb} \cdot b(C(\eb)) \\
				&= b(x) + \delta \cdot \left(\b_e - \frac{\b_e}{\b_\eb} \cdot \b_\eb \right) = b(x) = B.
	\end{align*}
	Moreover, it holds that
	\begin{align*}
		c(x')	&= c(x) + \delta \cdot c(\chi(C(e))) \\
				&\hspace{0.5cm} - \delta \cdot \frac{\b_e}{\b_\eb} \cdot c(\chi(C(\eb))) \\
				&= c(x) + \delta \cdot c(C(e)) - \delta \cdot \frac{\b_e}{\b_\eb} \cdot c(C(\eb)) \\
				&= c(x) + \delta \cdot \left(\c_e - \frac{\b_e}{\b_\eb} \cdot \c_\eb \right) \\
				&= c(x) + \underbrace{\delta}_{\geq 0} \cdot \underbrace{\d_e}_{< 0} \leq c(x). 
	\end{align*}
	By sending a small amount of $\delta \geq 0$~units of flow on $C(e)$ and $-\frac{\b_e}{\b_\eb} \cdot \delta$~units of flow on $C(\eb)$, we do not increase the objective value while maintaining feasibility. In fact, if we can choose a positive value for $\delta$, the objective value strictly decreases. Let $\theta_{e'} \colonequals (\chi(C(e)))_{e'} - \frac{\b_e}{\b_\eb} \cdot (\chi(C(\eb)))_{e'}$ denote the effect that an augmentation of one unit of flow on $C(e)$ and $-\frac{\b_e}{\b_\eb}$~units of flow on $C(\eb)$ has on edge~$e' \in E$. Moreover, let $\delta$ be defined as $\delta \colonequals \min_{e' \in E} \delta_{e'}$ with
	\begin{align*}
		\delta_{e'} &\colonequals \begin{cases}
			-\frac{x_{e'}}{\theta_{e'}} & \text{if } \theta_{e'} < 0, \\
			\frac{u_{e'} - x_{e'}}{\theta_{e'}} & \text{if } \theta_{e'} > 0, \\
			+\infty & \text{else}.
		\end{cases}
	\end{align*}
	Hence, by sending $\delta$ units of flow on $C(e)$ and $-\frac{\b_e}{\b_\eb}$~units of flow on $C(\eb)$, we maintain feasibility of the flow. Moreover, by the definition of $\delta$, there are several \emph{blocking edges}~$e'$ contained in $C(e)$ or $C(\eb)$ (or both) that fulfill~$\delta_{e'} = \delta$. We choose one of these blocking edges as the \emph{leaving edge}~$e'$, which consequently fulfills $x'_{e'} = 0$ or $x'_{e'} = u_{e'}$. We distinguish three cases:
	\begin{itemize}
		\item If $e' = e$, we can simply remove $e$ from $L$ and assign it to $U$. The new basis structure~$(L',T',U',\eb') \colonequals (L \setminus \{e\}, T, U \cup \{e\}, \eb)$ is then feasible again.
		\item If $e'=\eb$, we obtain a new basis structure by setting $(L',T',U',\eb') \colonequals (L \cup \{e'\} \setminus \{e\}, T, U, e)$ or $(L',T',U',\eb') \colonequals (L \setminus \{e\}, T, U \cup \{e'\}, e)$, depending on whether $x'_{e'} = 0$ or $x'_{e'} = u_{e'}$, respectively.
		\item Otherwise, we remove $e$ from $L$ and assign it to $T$. Moreover, we remove $e'$ from $T$ and assign it to $L$ or $U$, depending on whether $x'_{e'} = 0$ or $x'_{e'} = u_{e'}$, respectively, which yields the new basis structure~$(L',T',U',\eb') \colonequals (L \cup \{e'\} \setminus \{e\}, T \cup \{e\} \setminus \{e'\}, U, \eb)$ or $(L',T',U',\eb') \colonequals (L \setminus \{e\}, T \cup \{e\} \setminus \{e'\}, U \cup \{e'\}, \eb)$, respectively. Furthermore, for the case that $\eb$ is no longer contained in a cycle in $T' \cup \{\eb\}$, we assign $\eb$ to $T'$, remove $e$ from $T'$, and set $\eb' \colonequals e$.
	\end{itemize}
	In any case, we maintain a spanning tree~$T$ and ensure that $\eb$ closes a cycle with the edges in $T$. As in the traditional network simplex algorithm, we refer to such a step as a \emph{simplex pivot}. This pivot step is called \emph{degenerate} if $\delta=0$ and \emph{non-degenerate} else. In the former case, we refer to those edges with $\delta_e = 0$ as \emph{degenerate edges}. Note that the objective function strictly decreases only in the case of a non-degenerate simplex pivot.

	The case that there is an edge~$e \in U$ with $\d_e > 0$ is similar to the above case. By decreasing the flow on $C(e)$ by $\delta$~units and increasing the flow on $C(\eb)$ by $\frac{\b_e}{\b_\eb}$~units, i.e., by setting
	\begin{align}
		x' \colonequals x - \delta \cdot \chi(C(e)) + \delta \cdot \frac{\b_e}{\b_\eb} \cdot \chi(C(\eb)),   \label{eqn:CMCFPC:Simplex:SimplexUpdateU}
	\end{align}
	we maintain feasibility and improve the objective function for the case that $\delta > 0$.

	In the above discussion, we have assumed that $(L,T,U,\eb)$ is a (feasible) basis structure, which includes that $\b_\eb \neq 0$, i.e., that the usage fee of the cycle~$C(\eb)$ is non-zero. As it turns out, the usage fee of $C(\eb)$ remains non-zero after a simplex pivot as long as it was non-zero before the step, as it is shown in the following lemma:

	\begin{lemma}\label{lem:CMCFPC:Simplex:NoZeroCycle}
		Assume that $(L,T,U,\eb)$ is a feasible basis structure and let $\pi$ and $\mu$ denote the corresponding node potentials. Then the tuple~$(L',T',U',\eb')$ that results from a simplex pivot is a feasible basis structure again.
	\end{lemma}

	\begin{proof}
		Let $e$ denote the entering edge in the simplex pivot that leads to the new basis structure~$(L',T',U',\eb')$. As it was shown above, the flow that is induced by this new basis structure is feasible, again. Now assume for the sake of contradiction that $b^{\mu'}_{\o{e}'} = 0$.

		First, consider the case that the two cycles~$C(e)$ and $C(\eb)$ are edge-disjoint. Clearly, since either one of the edges in $C(e)$ or one of the edges in $C(\eb)$ becomes the leaving edge, one of the two cycles remains in $T' \cup \{\eb'\}$. Since $b(C(\eb)) = \b_\eb \neq 0$ by assumption, it must hold that $b(C(e)) = \b_e = 0$. However, in this case, the flow on the cycle~$C(\eb)$ does not change according to equations~\eqref{eqn:CMCFPC:Simplex:SimplexUpdateL} and \eqref{eqn:CMCFPC:Simplex:SimplexUpdateU}, i.e., the leaving edge must be contained in $C(e)$ and the cycle~$C(\eb)$ remains in $T' \cup \{\eb'\}$, which contradicts the assumption that $b^{\mu'}_{\o{e}'} = 0$.

		Now assume that the two cycles~$C(e)$ and $C(\eb)$ are not edge-disjoint. It is easy to see, that there is exactly one (undirected) simple path~$P_0$ that is contained in both $C(e)$ and $C(\eb)$ and that, consequently, contains neither $e$ nor $\eb$ (cf. \citep{CalicskanNetworkSimplex}). The leaving edge~$e'$ must then be contained in $P_0$, which can be seen as follows: For the case that $\b_e \neq 0$, none of the two cycles~$C(e)$ and $C(\eb)$ can still be contained in $T' \cup \{\eb'\}$ (otherwise, it would again hold that $b^{\mu'}_{\o{e}'} \neq 0$), i.e., the leaving edge must be a common edge of the two cycles, which lies on $P_0$. Else, if $\b_e = 0$, the leaving edge~$e'$ must be contained in $C(e)$ since the flow does not change on $C(\eb)$ as shown above. If $e'$ was not contained in $C(\eb)$ as well, the cycle~$C(\eb)$ would continue to exist in $T' \cup \{\eb'\}$, which, again, contradicts the assumption that $b^{\mu'}_{\o{e}'} = 0$. Hence, we obtain that $e' \in P_0$. Let $v$ and $w$ denote the end nodes of $P_0$ and let $P(e)$ and $P(\eb)$ denote the unique paths that connect $w$ with $v$ in $C(e) \setminus P_0$ and $C(\eb) \setminus P_0$, respectively. For some fixed orientation of $P(e)$ and $P(\eb)$, the new cycle~$C(\eb')$ is the concatenation of $P(e)$ and the reversal of $P(\eb)$ (cf. Figure~\ref{fig:CMCFPC:Simplex:NoZeroCycle}). Since $b^{\mu'}_{\o{e}'} = 0$ by assumption, it holds that $b(P(e)) = b(P(\eb))$, which implies that
		\begin{align*}
			\b_e &= b(C(e)) = b(P(e)) + b(P_0) \\
			&= b(P(\eb)) + b(P_0) = b(C(\eb)) = \b_\eb.
		\end{align*}
		However, according to equations~\eqref{eqn:CMCFPC:Simplex:SimplexUpdateL} and \eqref{eqn:CMCFPC:Simplex:SimplexUpdateU}, this implies that the flow on $x_{e'}$ remains unchanged, which contradicts the assumption that $e'$ is the leaving edge.
		\begin{figure}[ht!]
			\centering
			\begin{tikzpicture}[>=stealth',->]
				\draw[->,thick,xscale=1.5] (0,2) arc(45:307:1.4);
				\draw[->,thick,xscale=-1.5] (0,2) arc(45:307:1.4);
				\node[knoten] (V) at (0,0) {$v$};
				\node[knoten] (W) at (0,2) {$w$};
				\path[decoration={zigzag,amplitude=0.8mm,segment length=3mm}] (V) edge[decorate] node[right] {$P_0$} (W);
				\node at (-1.5,1) {$C(\eb)$};
				\node at (1.5,1) {$C(e)$};
				\node at (-2.5,-0.6) {$P(\eb)$};
				\node at (2.5,-0.6) {$P(e)$};
			\end{tikzpicture}
			\caption[The situation if the two cycles~$C(e)$ and $C(\eb)$ are not edge-disjoint.]{The situation if the two cycles~$C(e)$ and $C(\eb)$ are not edge-disjoint. Since the leaving edge~$e'$ lies on the path~$P_0$, the resulting cycle (thick) is the concatenation of $P(\eb)$ and $P(e)$.}
			\label{fig:CMCFPC:Simplex:NoZeroCycle}
		\end{figure}
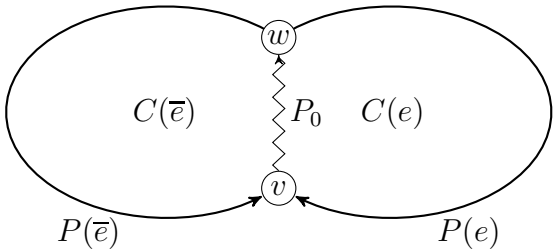
	\qed\end{proof}

	\section{Optimality Conditions}
	\label{sec:CMCFPC:Simplex:Optimality}

	The above discussion shows that, whenever we encounter an edge~$e \in L$ with $\d_e < 0$ or an edge~$e \in U$ with $\d_e > 0$, we can perform a simplex pivot and improve the objective function value (in the case that the pivot is non-degenerate). Conversely, as it turns out, whenever there are no such edges in a feasible basis structure, the corresponding feasible basic flow is an optimal solution to the given instance of \CMCFPC. In order to prove this result, we need the following lemma:

	\begin{lemma}\label{lem:CMCFPC:Simplex:OptimizationEquivalent}
		Let $(\pi,\mu)$ denote the node potentials corresponding to the basis structure~$(L,T,U,\eb)$. Any flow~$\o{x}$ with $b(\o{x}) = B$ is optimal if and only if it is optimal with respect to the objective function~$\d(x) \colonequals \sum_{e \in E} \d_e \cdot x_e$.
	\end{lemma}

	\begin{proof}
		Using equation~\eqref{eqn:CMCFPC:Simplex:DefReducedCosts} and the fact that $\sum_{e \in E} \c_e \cdot \o{x}_e = \sum_{e \in E} c_e \cdot \o{x}_e$ and $\sum_{e \in E} \b_e \cdot \o{x}_e = \sum_{e \in E} b_e \cdot \o{x}_e$ (cf., e.g., \citep[pp. 43--44]{Ahuja}), we get that the flow~$\o{x}$ fulfills the following property:
		\begin{align}
			\sum_{e \in E} \d_e \cdot \o{x}_e &= \sum_{e \in E} \left( \c_e - \c_\eb \cdot \frac{\b_e}{\b_\eb} \right) \cdot \o{x}_e \nonumber\\
			&= \sum_{e \in E} \c_e \cdot \o{x}_e - \frac{\c_\eb}{\b_\eb} \cdot \sum_{e \in E} \b_e \cdot \o{x}_e \nonumber \\
			&= \sum_{e \in E} c_e \cdot \o{x}_e - \frac{\c_\eb}{\b_\eb} \cdot \sum_{e \in E} b_e \cdot \o{x}_e \nonumber\\
			&= \sum_{e \in E} c_e \cdot \o{x}_e - \frac{\c_\eb}{\b_\eb} \cdot B. \label{eqn:CMCFPC:Simplex:OptimizationEquivalent}
		\end{align}
		Note that the value~$- \frac{\c_\eb}{\b_\eb} \cdot B$ only depends on the basis structure and is independent from the flow~$\o{x}$. Hence, for any feasible flow~$\o{x}$ with $b(\o{x}) = B$, the objective function values~$\d(\o{x})$ and $c(\o{x})$ only differ by a constant additive value, which shows the claim.
	\qed\end{proof}

	\begin{theorem}
		For a feasible basis structure~$(L,T,U,\eb)$ and the corresponding node potentials~$\pi$ and $\mu$, assume that the reduced costs~$\d$ fulfill the following conditions:
		{\begin{subequations}
		\begin{align}
			\d_e \geq 0 & \text{ for all } e \in L, \\
			\d_e = 0 & \text{ for all } e \in T \cup \{\eb\}, \\
			\d_e \leq 0 & \text{ for all } e \in U.
		\end{align}
		\end{subequations}}%
		Then the corresponding basic feasible flow~$x^*$ is optimal.
	\end{theorem}

	\begin{proof}
		Let $\d$ and $x^*$ be defined as above and let $\o{x}$ denote some arbitrary feasible flow. As shown in Lemma~\ref{lem:CMCFPC:Simplex:OptimizationEquivalent}, minimizing $c(\o{x}) = \sum_{e \in E} c_e \cdot \o{x}_e$ is equivalent to minimizing $\d(\o{x}) = \sum_{e \in E} \d_e \cdot \o{x}_e$. Since we have
		\begin{align*}
			\d(\o{x}) &= \sum_{e \in E} \d_e \cdot \o{x}_e \\
			&= \sum_{e \in L} \d_e \cdot \o{x}_e + \sum_{e \in T \cup \{e\}} \d_e \cdot \o{x}_e \\
			&\ \ \ + \sum_{e \in U} \d_e \cdot \o{x}_e \\
			&= \sum_{e \in L} \underbrace{\d_e}_{\geq 0} \cdot \underbrace{\o{x}_e}_{\geq 0} + \sum_{e \in U} \underbrace{\d_e}_{\leq 0} \cdot \underbrace{\o{x}_e}_{\leq u_e} \\
			&\geq \sum_{e \in U} \d_e \cdot u_e = \d(x^*),
		\end{align*}
		this shows that $x^*$ is optimal.
	\qed\end{proof}

	\section{Termination and Running Time}
	\label{sec:CMCFPC:Simplex:Termination}

	As shown above, we only make progress with respect to the objective function value if the corresponding simplex pivot is non-degenerate. However, like in the case of the traditional simplex method and the traditional network simplex algorithm (cf., e.g., \citep{Ahuja,DantzigGeneralizedFlows}), it may be possible to end in an infinite loop of degenerate pivots if no further steps are undertaken. In the case of the traditional network simplex algorithm, there are two common methods to prevent cycling of the procedure: One can either use a perturbed problem, in which the right-hand side vector of the LP formulation is suitably transformed, or use the concept of \emph{strongly feasible basis structures} in combination with a special leaving edge rule (cf. \citep{AhujaPaper}). In this section, we show that a combination of both approaches leads to a finite network simplex algorithm with pseudo-polynomial running time for \CMCFPC.

	In the following, we consider what we call the \emph{transformed problem} of the given instance of \CMCFPC, in which we replace the (previously integral) budget~$B$ by $B' \colonequals B + \frac{1}{2}$. In doing so, we maintain feasibility of the problem: According to Assumption~\ref{ass:CMCFPC:Simplex:InfeasibleMinCostFlow}, the minimum cost flow~$x$ that we obtain by some minimum cost flow algorithm fulfills $b(x) \geq B+1$. However, this implies that we can scale down $x$ to a feasible flow~$x'$ with $b(x') = B'$, so we can restrict our considerations to the transformed problem. As it turns out, each basic feasible flow of the transformed problem fulfills a useful property that will be essential throughout this section:

	\begin{lemma}\label{lem:CMCFPC:Simplex:IfOnCycleThenFractional}
		For each basis structure~$(L,T,U,\eb)$ of the transformed problem and its corresponding basic feasible flow~$x$, it holds that $x_e \notin \mathbb{N}_{\geq 0}$ for all $e \in C(\eb)$.
	\end{lemma}

	\begin{proof}
		According to Corollary~\ref{cor:CMCFPC:Simplex:IfFractionalThenOnCycle}, the flow~$x$ can be decomposed into an integral flow~$x^I$ and a flow~$x^C$ that is positive only on $C(\eb)$. Since $b_e \in \mathbb{N}_{\geq 0}$ for each $e \in E$, it holds that $\sum_{e \in E} b_e \cdot x^I_e \in \mathbb{N}_{\geq 0}$, so $b(x) = B' = B + \frac{1}{2}$ implies that $\sum_{e \in C(\eb)} b_e \cdot x^C_e = k + \frac{1}{2}$ for some integer~$k$. Since $x^C$ is positive only on $C(\eb)$, it holds that $x^C_e = \lambda$ for each $e \in C^+(\eb)$ and $x^C_e = -\lambda$ for each $e \in C^-(\eb)$ with $\lambda = \frac{k + \frac{1}{2}}{b(C(\eb))} \notin \mathbb{N}_{\geq 0}$, which shows the claim.
	\qed\end{proof}

	We now show that we can restrict our considerations solely to the transformed problem since an optimal basic solution that is obtained by an application of the network simplex algorithm to the transformed problem also yields an optimal basic solution of the original problem:

	\begin{lemma}
		An optimal basic solution of the transformed problem can be turned into an optimal basic solution of the original problem in $\mathcal{O}(n)$~time.
	\end{lemma}

	\begin{proof}
		Let $(L,T,U,\eb)$ denote a basis structure that implies an optimal solution~$x^*$ of the transformed problem. According to Lemma~\ref{lem:CMCFPC:Simplex:IfOnCycleThenFractional}, the flow~$x^*_e$ on each edge~$e \in C(\eb)$ fulfills~$x^*_e \notin \mathbb{N}_{\geq 0}$. In particular, this implies that $x^*_e \in (0,u_e)$ for each $e \in C(\eb)$, so we can increase or reduce the flow on the cycle by a small amount without violating any flow bounds. Since $x^*_e \in \mathbb{N}_{\geq 0}$ for each $e \notin C(\eb)$ according to Corollary~\ref{cor:CMCFPC:Simplex:IfFractionalThenOnCycle} and since $b_e \in \mathbb{N}_{\geq 0}$ for each $e \in E$, it must hold that we can increase or reduce the flow on $C(\eb)$ by $\delta$~units such that $\delta \cdot b(C(\eb)) = -\frac{1}{2}$, i.e., such that we obtain a feasible flow for the original problem. Moreover, note that $\d_e = 0$ for each $e \in C(\eb)$, i.e., the flow still fulfills the optimality conditions from Lemma~\ref{lem:CMCFPC:Simplex:OptimizationEquivalent}. Since the flow values on the edges in $L \cup U$ remain unchanged, the resulting flow is the basic solution corresponding to the basis structure~$(L,T,U,\eb)$ for the original problem, which shows the claim.
	\qed\end{proof}

	As noted above, one method to prevent cycling of the traditional network simplex algorithm is to use the concept of \emph{strongly feasible basis structures}, which are feasible basis structures in which every edge in $L$ is directed to the root node and every edge in $U$ heads away from the root node (cf. \citep[p. 422]{Ahuja}). As shown by \citet{AhujaPaper}, an equivalent definition is that, in the corresponding basic feasible flow, it is possible to send a positive amount of additional flow from every node~$v \in V$ to the root node via tree edges. For \CMCFPC, it turns out that a strongly feasible basis structure remains strongly feasible after a simplex pivot if the leaving edge is chosen appropriately, just as in the case of the traditional network simplex algorithm:

	\begin{lemma}\label{lem:CMCFPC:Simplex:StronglyFeasible}
		Let $(L,T,U,\eb)$ denote a strongly feasible basis structure of the transformed problem. The leaving edge~$e'$ can be chosen such that the basis structure~$(L',T',U',\eb')$ that results from a simplex pivot is again strongly feasible.
	\end{lemma}

	\begin{proof}
		Let $e=(v,w)$ denote the entering edge (we assume that $e \in L$; the case that $e \in U$ works analogously) and let $E' \subseteq T \cup \{\eb\}$ denote the set of blocking edges that determine the value of $\delta$ in the simplex pivot. Note that the graph that is induced by $T \cup \{\eb, e\}$ contains up to three simple cycles, one of which must carry a fractional amount of flow after the simplex pivot according to Lemma~\ref{lem:CMCFPC:Simplex:IfOnCycleThenFractional}. Hence, since the cycle that carries a fractional amount of flow cannot contain a blocking edge, there are three cases to distinguish: It either holds that $E' \subseteq C(e) \setminus C(\eb)$ or that $E' \subseteq C(\eb) \setminus C(e)$ or that $E' \subseteq C(e) \cap C(\eb)$ (cf. Figure~\ref{fig:CMCFPC:Simplex:StronglyFeasibleCycles}). We distinguish these three cases in the following. Note that, as in the proof of Lemma~\ref{lem:CMCFPC:Simplex:NoZeroCycle}, we get that $C(e) \cap C(\eb)$ corresponds to a single simple path~$P_0$ consisting of edges in $T$.

		\begin{figure}[ht!]
			\centering
			\begin{tikzpicture}[>=stealth',->] \node[scale=1.0] {\begin{tikzpicture}
			\draw[->,xscale=1.5] (0,2) arc(45:315:1.4);
			\draw[->,xscale=-1.5] (0,2) arc(45:315:1.4);
			\node[knoten] (A) at (0,0) {};
			\node[knoten] (B) at (0,2) {};
			\path[decoration={zigzag,amplitude=0.8mm,segment length=3mm}] (A) edge node[right] {$P_0$} (B);
			\node at (-1.5,1) {$C(\eb)$};
			\node at (1.5,1) {$C(e)$};

			\node at (2,-2) {\small $E' \subseteq C(e) \setminus C(\eb)$};
			\node[scale=0.6] at (-2,-2) {
				\begin{tikzpicture}[>=stealth',->]
				\draw[->,xscale=1.5,gray] (0,2) arc(45:315:1.4);
				\draw[->,xscale=-1.5,thick] (0,2) arc(45:315:1.4);
				\node[knoten] (A) at (0,0) {};
				\node[knoten] (B) at (0,2) {};
				\path (A) edge[gray] node[right] {$P_0$} (B);
				\node[gray] at (-1.5,1) {$C(\eb)$};
				\node at (1.5,1) {$C(e)$};
				\end{tikzpicture}
			};
			\node at (2,-4) {\small $E' \subseteq C(\eb) \setminus C(e)$};
			\node[scale=0.6] at (-2,-4) {
				\begin{tikzpicture}[>=stealth',->]
				\draw[->,xscale=1.5,thick] (0,2) arc(45:315:1.4);
				\draw[->,xscale=-1.5,gray] (0,2) arc(45:315:1.4);
				\node[knoten] (A) at (0,0) {};
				\node[knoten] (B) at (0,2) {};
				\path (A) edge[gray] node[right] {$P_0$} (B);
				\node at (-1.5,1) {$C(\eb)$};
				\node[gray] at (1.5,1) {$C(e)$};
				\end{tikzpicture}
			};
			\node at (2,-6) {\small $E' \subseteq C(e) \cap C(\eb)$};
			\node[scale=0.6] at (-2,-6) {
				\begin{tikzpicture}[>=stealth',->]
				\draw[->,xscale=1.5,gray] (0,2) arc(45:315:1.4);
				\draw[->,xscale=-1.5,gray] (0,2) arc(45:315:1.4);
				\node[knoten] (A) at (0,0) {};
				\node[knoten] (B) at (0,2) {};
				\path (A) edge[thick] node[right] {$P_0$} (B);
				\node[gray] at (-1.5,1) {$C(\eb)$};
				\node[gray] at (1.5,1) {$C(e)$};
				\end{tikzpicture}
			};
			\end{tikzpicture}};\end{tikzpicture}

			\caption{The graph induced by $T \cup \{e,\eb\}$ (top) and the three possible cases for the set~$E'$ of blocking edges (bottom). In each of these cases, the set~$E'$ is contained in the solid black segment.}
			\label{fig:CMCFPC:Simplex:StronglyFeasibleCycles}
		\end{figure}
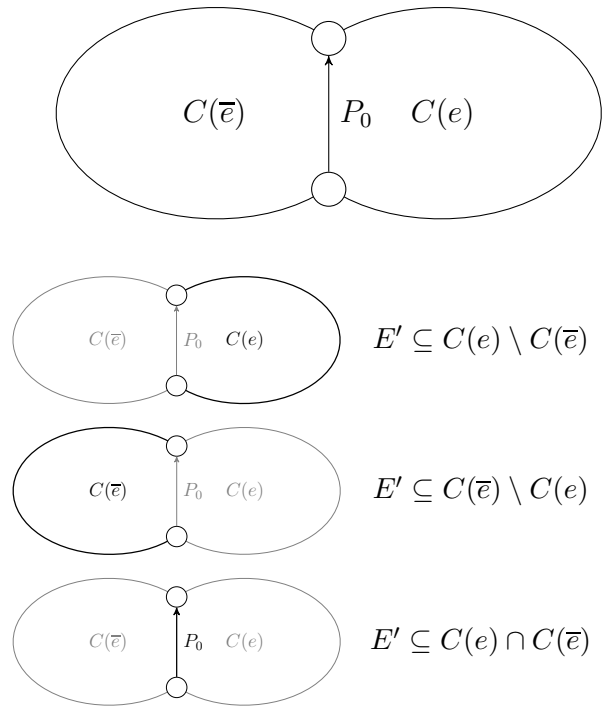

		First assume that $E' \subseteq C(e) \setminus C(\eb)$. In this case, it holds that $\eb' = \eb$ and the cycle~$C(\eb)$ continues to exist in $T' \cup \{\eb'\}$. Hence, the flow on all edges in $C(\eb)$ remains fractional and we are still able to send a positive amount of flow from any node in $C(\eb)$ to the apex\footnote{The \emph{apex} of a cycle~$C(e)$ with $e=(v,w)$ is the first common node of the two unique paths in $T$ from $v$ to the root and from $w$ to the root.} of $C(\eb)$. The rest of the proof of this case is analogous to the one for the traditional network simplex algorithm (cf., e.g., \citep[pp. 424--425]{Ahuja} and Figure~\ref{fig:CMCFPC:Simplex:StronglyFeasibleCase1}): We choose the leaving edge~$e' = (v',w')$ to be the last edge in $E'$ that occurs when traversing the cycle~$C(e)$ in the direction of $e$, starting at the apex~$z$ of $C(e)$. For the sake of notational simplicity, assume that $e' \in C^+(e)$ (the case that $e' \in C^-(e)$ works analogously). Clearly, if $\o{v}$ is any node on the path from $w'$ to $z$ (in the direction of the cycle), we can still send a positive amount of flow from $\o{v}$ to $z$ since there are no blocking edges on this path according to the choice of $e'$. On the other hand, for the case that the simplex pivot is non-degenerate, we send a positive amount of flow on the path from $z$ to $v'$ (which may be reduced by the flow on the cycle~$C(\eb)$ on the edges in $C(e) \cap C(\eb)$, but which will not be reduced to zero since $C(\eb)$ contains no blocking edges). Hence, we can send a positive amount of flow back from every node~$\o{v}$ on the path from $v'$ to $z$ in $T' \cup \{\eb'\} = T \cup \{e,\eb\} \setminus \{e'\}$. For the case that the simplex pivot is degenerate, it must hold that all blocking edges $E'$ lie on the path from $z$ to $v$ since $(L,T,U,\eb)$ is strongly feasible and we can, thus, send a positive amount of flow to $z$ on every edge on the path from $w$ to $z$. However, in the degenerate case, we do not change the flow on the edges on the path from $z$ to $v'$ and can, thus, still send a positive amount of flow from $v'$ to $z$. Note that the flow on every edge in $E \setminus (C(e) \cup C(\eb))$ does not change, so we can send a positive amount of flow from every node in $V$ to the root node~$r$ after the simplex pivot (possibly via the edges in $C(e) \setminus \{e'\}$). Hence, we maintain a strongly feasible basis structure in this case.

		\begin{figure}[ht!]
			\centering
			\begin{tikzpicture}[>=stealth',->]
			\node[knoten] (Z) at (0,0) {$z$};
			\node[knoten] (VP) at (-1, -2) {$v'$};
			\node[knoten] (WP) at (-1.5, -3) {$w'$};
			\node[knoten] (V) at (-2.5, -5) {$v$};
			\node[knoten] (W) at (2.5, -5) {$w$};

			\path[decoration={zigzag,amplitude=0.8mm,segment length=3mm}] (Z) edge[decorate] (VP);
			\path (VP) edge node[above left] {$e'$} (WP);
			\path[decoration={zigzag,amplitude=0.8mm,segment length=3mm}] (WP) edge[decorate] (V);
			\path (V) edge node[below] {$e$} (W);
			\path[decoration={zigzag,amplitude=0.8mm,segment length=3mm}] (W) edge[decorate] (Z);

			\draw[->] (0,-2.5) arc(90:405:.5);

			\end{tikzpicture}

			\caption{A cycle~$C(e)$ that is induced by the entering edge~$e=(v,w) \in L$ with $\d_e < 0$. After sending flow on $C(e)$, the leaving edge~$e'$ is the last blocking edge when traversing $C(e)$ in the direction of $e$ starting from the apex~$z$.}
			\label{fig:CMCFPC:Simplex:StronglyFeasibleCase1}
		\end{figure}
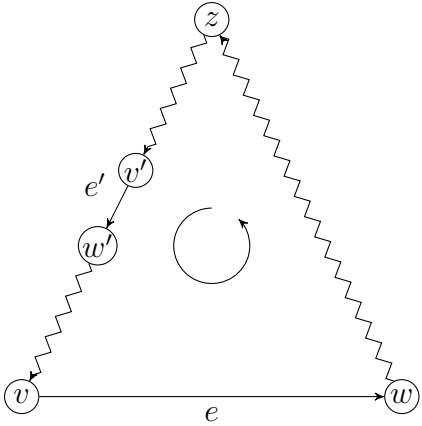

		The second case, in which $E' \subseteq C(\eb) \setminus C(e)$, works similar to the previous case. Note that we now get that $\eb' = e$ since the cycle~$C(\eb)$ vanishes. We choose the leaving edge to be the last blocking edge that occurs when traversing the cycle~$C(\eb)$ in the direction in that we send the flow in the simplex pivot, starting at the apex of $C(\eb)$. Note that the simplex pivot must be non-degenerate in this case since the blocking edges lie on $C(\eb)$ and every edge on $C(\eb)$ carries a fractional amount of flow before the simplex pivot.

		It remains to show that we maintain a strongly feasible basis structure in the case that $E' \subseteq P_0 = C(e) \cap C(\eb)$ (cf. Figure~\ref{fig:CMCFPC:Simplex:BlockingOnP0}). As in the previous case, the simplex pivot is non-degenerate since all edges in $C(\eb)$ carry a fractional amount of flow before the simplex pivot. Thus, the algorithm sends a positive amount of flow~$\delta$ along $C(e)$ and $\o{\delta} \colonequals -\frac{b(C(e))}{b(C(\eb))} \cdot \delta = -\frac{\b_e}{\b_\eb} \cdot \delta$~units of flow along $C(\eb)$. Since no edge in $C(e) \setminus C(\eb)$ and $C(\eb) \setminus C(e)$ is a blocking edge and we send flow on both cycles, neither of these edges is at its lower or upper bound after the simplex pivot (this also follows from the fact that the new cycle in $T' \cup \{\eb'\} = T \cup \{\eb,e\} \setminus \{e\}$ consists of the edges in $(C(e) \setminus C(\eb)) \cup (C(\eb) \setminus C(e))$ and, thus, carries a fractional amount of flow, cf. the gray paths in Figure~\ref{fig:CMCFPC:Simplex:BlockingOnP0}). Moreover, since $E' \subseteq P_0$ (and $E' \neq \emptyset$), we must have $\delta \neq \o{\delta}$. Thus, there is a unique direction in which the flow is sent on $P_0$ (from $z$ to $\o{w}$ in Figure~\ref{fig:CMCFPC:Simplex:BlockingOnP0}). We choose the leaving edge~$e'=(v',w')$ to be the last blocking edge that occurs on any of the two cycles when traversing the corresponding cycle in the direction of this flow, starting from the apex of the cycle. We can then send flow from $w'$ to the apexes of both cycles (since there are no further blocking edges on the corresponding subpath of $P_0$ and since the flow on the new cycle is fractional) and from $v'$ to the apexes (since we have sent a positive amount of flow to $v'$ on the corresponding subpath of $P_0$ and since the flow on the new cycle is fractional). Hence, using the same arguments as in the previous two cases, we get that we can send a positive amount of flow from each node~$v \in V$ to the root node, which concludes the proof.
		\begin{figure}[ht!]
			\centering
			\begin{tikzpicture}[>=stealth',->]
				\node[knoten] (Zb) at (0,0) {$\o{z}$};
				\node[knoten] (Z) at (-1,-1) {$z$};
				\node[knoten] (ebv) at (4,-4) {$\o{v}$};
				\node[knoten] (ebw) at (0.5,-4) {$\o{w}$};
				\node[knoten] (ew) at (1,-5) {$w$};
				\node[knoten] (ev) at (-4,-5) {$v$};
				\node[knoten] (evp) at (-0.25,-2.5) {$v'$};
				\node[knoten] (ewp) at (0.075,-3.25) {$w'$};

				\path[decoration={zigzag,amplitude=0.8mm,segment length=3mm}] (Zb) edge[decorate,gray] (Z);
				\path[decoration={zigzag,amplitude=0.8mm,segment length=3mm}] (Zb) edge[decorate,gray] (ebv);
				\path[decoration={zigzag,amplitude=0.8mm,segment length=3mm}] (Z) edge[decorate,gray] (ev);
				\path[decoration={zigzag,amplitude=0.8mm,segment length=3mm}] (ebw) edge[decorate,gray] (ew);
				\path[decoration={zigzag,amplitude=0.8mm,segment length=3mm}] (Z) edge[decorate,thick] (evp);
				\path[decoration={zigzag,amplitude=0.8mm,segment length=3mm}] (ewp) edge[decorate,thick] (ebw);
				\path (ebv) edge[gray] node[below,gray] {$\o{e}$} (ebw);
				\path (ev) edge[gray] node[below,gray] {$e$} (ew);
				\path (evp) edge[thick] node[left,thick] {$e'$} (ewp);

				\draw[->,gray] (-1.5,-3.25) arc(90:405:.5);
				\draw[->,gray] (1.25,-2.25) arc(90:405:.5);
				\draw[->] (-0.5,-1.25) -- (0.75,-3.5);
			\end{tikzpicture}
			\caption{The case that $E' \subseteq P_0 = C(e) \cap C(\eb)$. The leaving edge~$e'$ is chosen to be the last edge on $P_0$ when traversing any of the two cycles~$C(e)$ and $C(\eb)$ in the direction of the flow change on $P_0$.}
			\label{fig:CMCFPC:Simplex:BlockingOnP0}
		\end{figure}
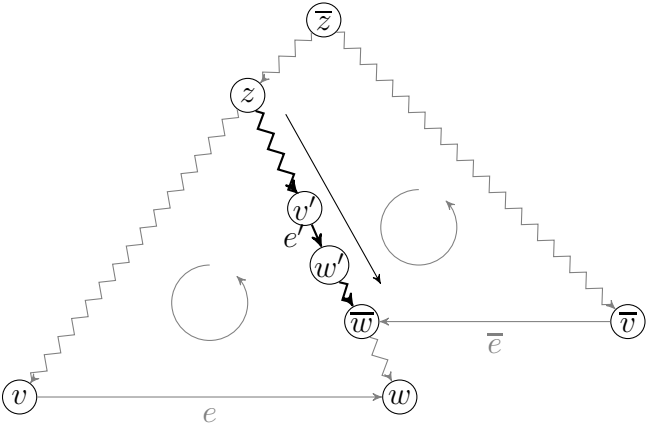
	\qed\end{proof}

	Lemma~\ref{lem:CMCFPC:Simplex:StronglyFeasible} builds the foundation for the following theorem, which shows that the network simplex algorithm for \CMCFPC does not cycle when using strongly feasible basis structures:

	\begin{theorem}\label{thm:CMCFPC:Simplex:DegeneratePivots}
		The network simplex algorithm applied to the transformed problem terminates within a finite number of simplex pivots when using strongly feasible basis structures. Moreover, the number of consecutive degenerate simplex pivots is bounded by $\mathcal{O}(n^3 \mathcal{CB})$.\footnote{In the following, we denote by $\mathcal{C}$, $\mathcal{U}$, and $\mathcal{B}$ the maximum absolute values of edge costs, capacities, and usage fees, respectively.} 
	\end{theorem}

	\begin{proof}
		Let $(L,T,U,\eb)$ denote a basis structure and $x$ denote its corresponding basic feasible flow. Consider a degenerate simplex pivot that occurs when adding the entering edge~$e=(v,w)$ to $T \cup \{\eb\}$ and choosing the leaving edge~$e'=(v',w')$ according to the leaving edge rules given in the proof of Lemma~\ref{lem:CMCFPC:Simplex:StronglyFeasible}, which leads to a new basis structure~$(L',T',U',\eb')$. Since the flow on $C(\eb)$ is fractional according to Lemma~\ref{lem:CMCFPC:Simplex:IfOnCycleThenFractional}, none of the edges on the cycle~$C(\eb)$ can be degenerate. Thus, it holds that $\eb' = \eb$ and that the cycle~$C(\eb)$ continues to exist in $T' \cup \{\eb'\}$. Moreover, as in the proof of Lemma~\ref{lem:CMCFPC:Simplex:StronglyFeasible}, the leaving edge~$e'$ must lie on the path in $T$ from the apex~$z$ of $C(e)$ to $v$ since the basis structure is strongly feasible (cf. Figure~\ref{fig:CMCFPC:Simplex:StronglyFeasibleCase1}). As in the traditional network simplex algorithm, after the simplex pivot, the node potentials~$\pi_v$ and $\mu_v$ are increased (decreased) by $\c_e$ and $\b_e$, respectively, for all nodes~$v$ on the path from $w'$ to $v$ in $T$ if $e \in L$ ($e \in U$), while the remaining node potentials remain unchanged (cf., e.g., \citep[p. 425]{Ahuja}). Thus, for each $v$ on this path, the new node potentials~$\pi'_v$ and $\mu'_v$ fulfill
		\begin{align*}
			& \pi'_v + \frac{\c_\eb}{\b_\eb} \cdot \mu'_v \\
			=\ & (\pi_v + \c_e) + \frac{\c_\eb}{\b_\eb} \cdot (\mu_v + \b_e) \\
			=\ & \left(\pi_v + \frac{\c_\eb}{\b_\eb} \cdot \mu_v \right) + \left(\c_e + \frac{\c_\eb}{\b_\eb} \cdot \b_e \right) \\
			=\ & \left(\pi_v + \frac{\c_\eb}{\b_\eb} \cdot \mu_v \right) + \d_e \\
			<\ & \pi_v + \frac{\c_\eb}{\b_\eb} \cdot \mu_v
		\end{align*}
		for the case that $e \in L$. Otherwise, if $e \in U$, we get that
		\begin{align*}
			& \pi'_v + \frac{\c_\eb}{\b_\eb} \cdot \mu'_v \\
			=\ & (\pi_v - \c_e) + \frac{\c_\eb}{\b_\eb} \cdot (\mu_v - \b_e) \\
			=\ & \left(\pi_v + \frac{\c_\eb}{\b_\eb} \cdot \mu_v \right) - \left(\c_e + \frac{\c_\eb}{\b_\eb} \cdot \b_e \right) \\
			=\ & \left(\pi_v + \frac{\c_\eb}{\b_\eb} \cdot \mu_v \right) - \d_e \\
			<\ & \pi_v + \frac{\c_\eb}{\b_\eb} \cdot \mu_v.
		\end{align*}
		Hence, the value~$\sum_{v \in V} \pi_v + \frac{\c_\eb}{\b_\eb} \cdot \mu_v$ decreases strictly after each degenerate simplex pivot. However, the values~$\pi_v$ are integers in $\{-n\mathcal{C}, \ldots, n\mathcal{C}\}$, while the values~$\mu_v$ lie in $\{-n\mathcal{B}, \ldots, n\mathcal{B}\}$ for each~$v \in V$ (cf. \citep[p. 425]{Ahuja}). Thus, since there are only $\mathcal{O}(n^2 \mathcal{CB})$ combinations of integral values for $\pi_v$ and $\mu_v$ for each node~$v \in V$, the algorithm performs a non-degenerate pivot after at most $\mathcal{O}(n^3 \mathcal{CB})$ degenerate simplex pivots and the claim follows.
	\qed\end{proof}

	While the leaving edge rules as described above guarantee finiteness of the procedure, we can reduce the number of non-degenerate simplex pivots by applying Dantzig's pivoting rule (cf., e.g., \citep{AhujaPaper}), i.e., by choosing the entering edge to be the one with the largest violation of its optimality conditions:

	\begin{lemma}\label{lem:CMCFPC:Simplex:NondegeneratePivots}
		The network simplex algorithm applied to the transformed problem performs $\mathcal{O}(n m \mathcal{UB} \log(m\mathcal{CUB}))$ non-degenerate simplex pivots in total when using Dantzig's pivoting rule.
	\end{lemma}

	\begin{proof}
		The proof of the lemma is similar to the one for the traditional network simplex algorithm given in \citep{AhujaPaper}. Let $x^k$ denote the basic feasible flow that is obtained after the $k$-th non-degenerate step of the algorithm and let $c(x^k)$ denote its objective function value. Moreover, let $(L,T,U,\eb)$ denote the corresponding basis structure. According to Corollary~\ref{cor:CMCFPC:Simplex:IfFractionalThenOnCycle}, each flow~$x^k$ can be decomposed into an integral flow~$x^I$ and a fractional flow~$x^C$ on the cycle~$C(\eb)$. In particular, since~$x^k$ satisfies $b(x^k) = B' = B + \frac{1}{2}$ and $x^C$ is a flow on the cycle~$C(\eb)$, it holds that $x^C_e = \frac{p}{b(C(\eb))}$ for $p \colonequals (B + \frac{1}{2}) - \sum_{e \in E} b_e \cdot x^I_e$, so the flow on every edge in $x^k$ is an integral multiple of $\frac{1}{2 \cdot b(C(\eb))}$. Thus, it holds that $|x^k_e - x^{k+1}_e|$ is either zero or at least $\frac{1}{2n\mathcal{B}}$ for each edge~$e \in E$. Moreover, note that the minimum absolute value of the reduced costs of any edge~$e$ that violates its optimality condition can be bounded as follows:
		\begin{align*}
			|\d_e| &= \left| \c_e - \frac{\b_e}{\b_\eb} \cdot \c_\eb \right| = \left| \frac{\c_e \cdot \b_\eb - \b_e \cdot \c_\eb}{\b_\eb} \right| \\
			&\geq \frac{1}{b(C(\eb))} \geq \frac{1}{n\mathcal{B}}.
		\end{align*}
		Since the objective function value of any flow is bounded from below by $-m\mathcal{C}\mathcal{U}$, we, thus, get that the maximum number of non-degenerate simplex pivots \emph{without} using Dantzig's pivoting rule is bounded by $\mathcal{O}(m\mathcal{CU} \cdot 2n\mathcal{B} \cdot n\mathcal{B}) = \mathcal{O}(n^2m\mathcal{CUB}^2)$.

		Let $\Delta \colonequals \max\{ -\min_{e \in L} \d_e, \max_{e \in U} \d_e \}$ denote the maximum violation of the optimality conditions of any edge in $L \cup U$ and let $e$ denote the corresponding edge that is chosen based on Dantzig's pivoting rule. Since sending one unit of flow over $C(e)$ reduces the objective function value by $\Delta$, we get that
		\begin{align}
			c(x^k) - c(x^{k+1}) \geq \frac{\Delta}{2n\mathcal{B}}  \label{eqn:CMCFPC:Simplex:NondegeneratePivotsProgressLowerBound}
		\end{align}
		Moreover, if $x^*$ denotes an optimal solution to the problem, we get according to equation~\eqref{eqn:CMCFPC:Simplex:OptimizationEquivalent} in the proof of Lemma~\ref{lem:CMCFPC:Simplex:OptimizationEquivalent} that
		\begin{align}
			& c(x^k) - c(x^*) \nonumber\\
			=\ & \d(x^k) - \d(x^*) \nonumber\\
			=\ & \sum_{e \in E} \d_e \cdot (x^k_e - x^*_e) \nonumber\\
			=\ & \sum_{e \in L} \d_e \cdot (-x^*_e) + \sum_{e \in U} \d_e \cdot (u_e - x^*_e) \\
			\leq\ & m \Delta \mathcal{U}.    \label{eqn:CMCFPC:Simplex:NondegeneratePivotsProgressUpperBound}
		\end{align}
		Combining equations~\eqref{eqn:CMCFPC:Simplex:NondegeneratePivotsProgressLowerBound} and \eqref{eqn:CMCFPC:Simplex:NondegeneratePivotsProgressUpperBound}, we, thus, get that
		\begin{align*}
			c(x^k) - c(x^{k+1}) \geq \frac{c(x^k) - c(x^*)}{2nm \mathcal{UB}},
		\end{align*}
		i.e., after each non-degenerate simplex pivot, the gap to the optimal solution with respect to the objective function value is reduced by a factor of at least $\frac{1}{2nm \mathcal{UB}}$. \citet[p. 67]{Ahuja} show that, if $H$ is the maximum number of improving steps of any algorithm and if this algorithm reduces the gap to the optimal solution by a fraction of at least $\alpha$ in each step, then the maximum number of steps is bounded by $\mathcal{O}(\frac{1}{\alpha} \log H)$. Thus, since $H \in \mathcal{O}(n^2m\mathcal{CUB}^2)$ in our case as shown above, we get that the maximum number of non-degenerate simplex pivots using Dantzig's pivoting rule is in $\mathcal{O}(2nm \mathcal{UB} \log(n^2m\mathcal{CUB}^2)) = \mathcal{O}(nm \mathcal{UB} \log(m\mathcal{CUB}))$.
	\qed\end{proof}

	In Lemma~\ref{lem:CMCFPC:Simplex:StronglyFeasible}, it was shown that we can obtain a strongly feasible basis structure again when performing a simplex pivot on a strongly feasible basis structure. However, it remains open how to determine an initial strongly feasible basis structure. This will be shown in the following lemma:

	\begin{lemma}\label{lem:CMCFPC:Simplex:InitialSolution}
		An initial strongly feasible basis structure~$(L,T,U,\eb)$ for \CMCFPC and the corresponding basic feasible solution~$x$ can be determined in $\mathcal{O}(m)$~time.
	\end{lemma}

	\begin{proof}
		For two arbitrary nodes~$v,w \in V$ in the given graph~$G=(V,E)$, we insert an artificial edge~$e_0 = (v,w)$ with costs~$c_{e_0} \colonequals 1$, capacity~$u_{e_0} \colonequals 1$, and usage fees~$b_{e_0} \colonequals B$. Moreover, we insert a second artificial edge~$e'_0 = (w,v)$ with costs~$c_{e'_0} \colonequals 1$, capacity~$u_{e'_0} \colonequals 1$, and usage fees~$b_{e'_0} \colonequals B+1$. The initial basic feasible solution~$x$ is defined by $x_{e_0} \colonequals 0.5$, $x_{e'_0} \colonequals 0.5$, and $x_e \colonequals 0$ for each $e \in E$ such that $b(x) = 0.5 \cdot (2B + 1) = B'$. The spanning tree~$T$ consists of $e_0$ as well as a spanning tree of the nodes in $V \setminus \{v\}$ that is a directed in-tree\footnotemark\xspace with root~$w$. Note that such an in-tree exists according to Assumption~\ref{ass:CMCFPC:Simplex:StronglyConnected} and can be found, e.g., by a depth-first search in $\mathcal{O}(m)$~time. Moreover, note that we can send a positive amount of flow from every node in $V$ to $v$ by using the unique path in the in-tree in combination with $e_0$. Hence, we obtain a strongly feasible basis structure by setting $\eb \colonequals e'_0$, choosing $T$ as defined above, and setting~$U \colonequals \emptyset$ and $L \colonequals E \setminus T$. Note that the new edges do not influence the optimal solution value since $e_0$ and $e'_0$ will have zero flow in any optimal solution.
		\footnotetext{An \emph{in-tree} is a tree in which all edges are directed towards the root node.}
	\qed\end{proof}

	We are now ready to prove the main result of the paper:

	\begin{theorem}
		The network simplex algorithm for \CMCFPC can be implemented to run in $\mathcal{O}(n^4 m^2 \mathcal{CUB}^2 \cdot \log(m\mathcal{CUB}))$~time.
	\end{theorem}

	\begin{proof}
		According to Lemma~\ref{lem:CMCFPC:Simplex:InitialSolution}, we can determine an initial basis structure and the corresponding basic feasible solution in $\mathcal{O}(m)$~time. It is easy to see that a single simplex pivot as described above can be implemented to run in $\mathcal{O}(m)$~time, including the overhead to determine the entering edge and the leaving edge according to Dantzig's pivoting rule and the above leaving edge rules. Moreover, the maximum number of non-degenerate simplex pivots is given by $\mathcal{O}(nm \mathcal{UB} \log(m\mathcal{CUB}))$ as shown in Lemma~\ref{lem:CMCFPC:Simplex:NondegeneratePivots}. In the worst case, each of these non-degenerate simplex pivots is followed by a sequence of $\mathcal{O}(n^3 \mathcal{BC})$ degenerate pivots according to Theorem~\ref{thm:CMCFPC:Simplex:DegeneratePivots}, which leads to an overall running time of 
		\begin{align*}
			&\mathcal{O}(m \cdot n^3 \mathcal{BC} \cdot nm \mathcal{UB} \log(m\mathcal{CUB})) \\
			=\ & \mathcal{O}(n^4 m^2 \mathcal{CUB}^2 \log (m\mathcal{CUB})),
		\end{align*}
		which shows the claim.
	\qed\end{proof}

	\section{Computational Results}

	In order to evaluate the empirical performance of the presented algorithm, we implemented it in C++ and compiled it for Windows 7 64bit with Microsoft Visual Studio Community 2015 with all available optimization options. The tests were performed on an Intel Core i5 processor at 2.53 GHz with one core and 4GB of RAM. We implemented the presented algorithm using Dantzig's pivoting rule and evaluated its performance against Gurobi 6.5.1 64bit \citep{Gurobi}. As in \citep{CalicskanNetworkSimplex}, the underlying networks were constructed with NETGEN \citep{Netgen}. For densities~$d \in \{8,16,32\}$ and node sets of size $n \in \{2^i : i \in \{8,\dotsc,15\}\}$, we generated test instances with $n$~nodes and $m \colonequals n \cdot d$~edges and computed the mean execution times over ten instances each.

	As a validation check, we first adapted the approach of \citet{CalicskanNetworkSimplex} by using an in-tree containing the shortest paths from each node $v \in V \setminus \{t\}$ to some sink~$t$ with respect to the usage fees~$b_e$ as the starting solution (cf. Lemma~\ref{lem:CMCFPC:Simplex:InitialSolution}) and by applying the algorithm to the maximum flow variant of the problem. For this subproblem, we observed similar running times as stated in \citep{CalicskanNetworkSimplex}, beating Gurobi by factors of up to $368$. These remarkable running times result from the fact that the procedure only performs a very little number of iterations. One possible explanation is that an optimal solution to the budget-constrained maximum flow problem can be obtained by repeatedly augmenting flow on shortest paths with relation to the usage fees~$b_e$ in the residual network (cf. \citep{AhujaConstrainedMaxFlow}), which can be done within a few number of iterations due to the chosen starting solution.

	However, when applied to the minimum cost flow variant of the problem, the running times became significantly worse. We compared the presented network simplex algorithm with the dual, primal, and barrier solvers provided with Gurobi. As it is shown in Table~\ref{tab:CMCFPC:ResultsMCF}, the dual solver was the fastest among all three solvers in most instances. For small instances, our specialized network simplex algorithm could beat all solvers, but performed less competitive in larger instances and was slower than the dual solver by a factor of up to $4.29$. Nevertheless, our network simplex algorithm was faster than Gurobi's barrier solver in all cases and faster than the primal solver for instances with density $8$ and $16$. Independent of the achieved running times, one major advantage of the presented algorithm is that it only consumed 58MB of memory even in the largest instance while Gurobi's dual solver needed up to 2.71GB RAM and did not solve the problem when compiled at 32bit (while using a 32bit version of Gurobi) since only 2GB of memory can be addressed.

	While the number of non-degenerate pivots of the network simplex method was still low, the percentage of degenerate pivots amounted to a fraction of up to $98.9\%$. Using a more sophisticated starting solution based on a scaled (traditional) minimum cost flow that is turned into a basic feasible flow (cf. Theorem~\ref{thm:CMCFPC:Simplex:OptimalBasicFeasibleFlow}), the relative amount of degenerate pivots could be reduced to $17\%$ but the progress in each non-degenerate iteration became worse.

	\begin{table*}[ht]
		\centering
		\small
		\begin{tabular}{rrr|r|rrr}
			\multicolumn{3}{l|}{Network size} & \multicolumn{4}{l}{CPU times (seconds)} \\\hline
			$n$ & $d$ & $m$ &		Network simplex & \multicolumn{3}{l}{Gurobi} \\
			&&&						& \multicolumn{1}{l}{Dual} & \multicolumn{1}{l}{Primal} & \multicolumn{1}{l}{Barrier} \\\hline
			256 & 8 & 2048 &		0.0062 & 0.0252 & 0.0298 & 0.1206 \\
			512 & 8 & 4096 &		0.0208 & 0.0610 & 0.0618 & 0.3805 \\
			1024 & 8 & 8192 &		0.0741 & 0.1120 & 0.1820 & 1.0589 \\
			2048 & 8 & 16384 &		0.3602 & 0.4868 & 0.6541 & 5.0358 \\
			4096 & 8 & 32768 &		1.6589 & 1.3137 & 2.5439 & 42.9839 \\
			8192 & 8 & 65536 &		6.4174 & 4.6624 & 12.7127 & 224.0670 \\
			16384 & 8 & 131072 &	39.0361 & 23.7921 & 71.9850 & 1501.2466 \\
			32768 & 8 & 262144 &	232.0375 & 133.3492 & 478.4381 & --- \\
			&&&&&&\\
			256 & 16 & 4096 &		0.0127 & 0.0347 & 0.0414 & 0.2308 \\
			512 & 16 & 8192 &		0.0472 & 0.1254 & 0.0928 & 0.6435 \\
			1024 & 16 & 16384 &		0.2132 & 0.4578 & 0.2869 & 2.6172 \\
			2048 & 16 & 32768 &		0.7776 & 0.7839 & 0.9485 & 11.3058 \\
			4096 & 16 & 65536 &		3.8542 & 3.9262 & 3.9801 & 72.5990 \\
			8192 & 16 & 131072 &	19.2471 & 13.3666 & 19.6259 & 410.8815 \\
			16384 & 16 & 262144 &	114.8645 & 57.9140 & 114.9896 & --- \\
			32768 & 16 & 524288 &	553.8394 & 273.8601 & 656.0465 & --- \\
			&&&&&&\\
			256 & 32 & 8192 &		0.0319 & 0.0602 & 0.0722 & 0.5702 \\
			512 & 32 & 16384 &		0.1266 & 0.1909 & 0.1875 & 1.1407 \\
			1024 & 32 & 32768 &		0.5474 & 0.5659 & 0.4983 & 4.1651 \\
			2048 & 32 & 65536 &		2.5307 & 1.1637 & 1.4056 & 22.3124 \\
			4096 & 32 & 131072 &	11.9041 & 4.9066 & 5.4382 & 101.0424 \\
			8192 & 32 & 262144 &	64.0416 & 26.5862 & 26.6424 & 668.8592 \\
			16384 & 32 & 524288 &	358.6511 & 97.3421 & 125.7771 & --- \\
			32768 & 32 & 1048576 &	1740.6900 & 406.1623 & 761.5276 & --- \\
		\end{tabular}
		\caption{Computational results comparing the running time of the presented network simplex algorithm with Gurobi's dual, primal, and barrier solvers. Instances without running times did not obtain a solution within 1800 seconds.}
		\label{tab:CMCFPC:ResultsMCF}
	\end{table*}

	\section{Conclusion and Future Work}

	In this paper, we developed a specialized network simplex method for the budget-constrained minimum cost flow problem. In particular, we proved optimality criteria for the problem based on a novel incorporation of two kinds of integral node potentials and three kinds of reduced costs and presented a fully combinatorial description of the procedure. Moreover, we combined two common techniques that are used to prevent cycling in the traditional network simplex algorithm into a rule for the leaving edge that prevents cycling. Finally, we could show that Dantzig's pivoting rule can be used in order to reduce the number of non-degenerate pivots significantly, which in combination with a pseudo-polynomial number of successive degenerate pivots leads to a pseudo-polynomial time bound for the overall procedure.

	It remains open for future research if the practical performance of the procedure can be improved by using a more sophisticated starting solution similar to the case of the maximum flow variant of the problem. One key issue in this respect seems to develop new rules for the choice of the entering and the leaving edge, which reduce the number of degenerate iterations. Thus, at the time being, the results of this paper can be viewed more as a theoretical contribution, which establishes a pseudo-polynomial running time and closes open issues left, e.g., in \citep{CalicskanNetworkSimplex}.

	\section*{Acknowledgements}

	This work was partially supported by the German Federal Ministry of Education and Research within the project ``SinOptiKom -- Cross-sectoral Optimization of Transformation Processes in Municipal Infrastructures in Rural Areas''.
	\\

	\bibliographystyle{elsarticle-harv}
	\bibliography{/Users/holzhaus/Documents/Forschung/literature}
\end{document}